\documentclass[twocolumn,amsmath,amssymb,10pt,aps]{revtex4}
  
\usepackage[utf8]{inputenc}
\usepackage[T1]{fontenc}

\usepackage{amsmath}
\usepackage{amsfonts}
\usepackage{graphicx}
\usepackage{yfonts}
\usepackage{color}
\usepackage[normalem]{ulem}
\usepackage{amsthm}
\usepackage{bm}
\usepackage{bbm}
\usepackage{mathtools}
\usepackage{array}
\usepackage{placeins}
\usepackage{enumitem}

\newcommand{\der}{\,\mathrm{d}}
\newcommand\bigforall{\mbox{\Large $\mathsurround=1pt\forall$}}

\def\<{\langle}
\def\>{\rangle}

\newcommand{\Tr}{\mathrm{Tr}}
\def\oper{{\mathchoice{\rm 1\mskip-4mu l}{\rm 1\mskip-4mu l}
{\rm 1\mskip-4.5mu l}{\rm 1\mskip-5mu l}}}
\DeclareMathAlphabet\mathbfcal{OMS}{cmsy}{b}{n}

\mathchardef\mhyphen="2D 

\newtheorem{Theorem}{Theorem}
\newtheorem{Lemma}{Lemma}
\newtheorem{Corollary}{Corollary}

\newtheorem{Proposition}{Proposition}
\newtheorem{Example}{Example}

\begin{document}

\title{Non-Markovianity criteria for mixtures of noninvertible Pauli dynamical maps}

\author{Katarzyna Siudzi\'{n}ska}
\affiliation{Institute of Physics, Faculty of Physics, Astronomy and Informatics \\  Nicolaus Copernicus University, ul. Grudzi\k{a}dzka 5/7, 87--100 Toru\'{n}, Poland}

\begin{abstract}
We analyze the connections between the non-Markovianity degree of the most general phase-damping qubit maps and their legitimate mixtures. Using the results for image non-increasing dynamical maps, we formulate the necessary and sufficient conditions for the Pauli maps to satisfy specific divisibility criteria. Next, we examine how the non-Markovianity properties for (in general noninvertible) Pauli dynamical maps influence the properties of their convex combinations. Our results are illustrated with instructive examples. For P-divisible maps, we propose a legitimate time-local generator whose all decoherence rates are temporarily infinite.
\end{abstract}

\flushbottom

\maketitle

\thispagestyle{empty}

\section{Introduction}

Quantum evolution with memory effects is a rapidly developing research area \cite{BLPV,Alonso} due to its applications in quantum information processing and
quantum communication \cite{BreuerPetr,RivasHuelga,Nielsen}.
Modern experimental methods make it possible to observe new non-Markovian effects caused by the interactions in the environment \cite{Bernardes,JinGiovannetti,WuHou,Walborn}.
One can even quantify \cite{Balthazar} and control \cite{LiuHuang,YuWang} the degree of non-Markovianity through environmental states manipulation.
Hence, it is important to further develop the theory of open quantum system dynamics that goes beyond the Markovian regime with emphasis on characterization and quantification \cite{RHP2,Wiseman}. Two main approaches are based on the divisibility \cite{RHP} and information backflow \cite{BLP}.

The evolution of open quantum systems is provided by dynamical maps, which are time-parameterized families of quantum channels (completely positive, trace-preserving maps) $\{\Lambda(t):\,t\geq 0,\,\Lambda(0)=\oper\}$ acting on the space $\mathcal{B}(\mathcal{H})$ of bounded operators on the Hilbert space $\mathcal{H}$. For an arbitrary initial state $\rho$, the evolved state is given by $\rho(t)=\Lambda(t)[\rho]$. It is evident that all the information about quantum dynamics, including its Markovianity, is encoded in the properties of $\Lambda(t)$. In one of the most popular approaches, the Markovianity degree of the evolution is related to the divisibility of the associated dynamical map \cite{RHP,Wolf}. Recall that $\Lambda(t)$ is divisible if it can be decomposed into
\begin{equation}\label{div}
\Lambda(t)=V(t,s)\Lambda(s)
\end{equation}
at all times $t\geq s\geq 0$, where $V(t,s)$ is a trace-preserving propagator. Now, if $V(t,s)$ is completely positive, then the corresponding $\Lambda(t)$ is CP-divisible, and hence the evolution it provides is called {\it Markovian} \cite{RHP,Wolf}. If the complete positivity of the propagator is broken, the evolution becomes non-Markovian. There exists a hierarchy of $k$-positive maps $V(t,s)$ that correspond to $k$-divisible dynamical maps \cite{Sabrina}. For $k=d$, where $d$ is the dimension of the underlying Hilbert space, one recovers CP-divisible $\Lambda(t)$. On the other hand, $k=1$ reproduces P-divisible $\Lambda(t)$, for which $V(t,s)$ is a positive map. The property of P-divisibility relates to the classical definition of Markovianity for one-point
probabilities in stochastic processes \cite{Piilo}. Further distinctions exist, where the positive but not completely positive (PnCP) $V(t,s)$ is associated with PnCP-divisible dynamical maps. Such evolution is referred to as {\it weakly non-Markovian} \cite{Sabrina}. Every P-divisible evolution satisfies the Breuer-Laine-Piilo condition for no information backflow from the system to the environment \cite{BLP}. Moreover, an {\it essentially non-Markovian} evolution arises from the dynamical maps that are not even P-divisible \cite{Filip2}.

If $\Lambda(t)$ is invertible, then the propagator of quantum evolution is well defined, as eq. (\ref{div}) gives $V(t,s)=\Lambda(t)\Lambda^{-1}(s)$. Moreover, in this case, the CP-divisibility of the dynamical map is closely related to the properties of the time-local generator $\mathcal{L}(t)$ that enters the time-local master equation
\begin{equation}
\frac{\der}{\der t}\Lambda(t)=\mathcal{L}(t)\Lambda(t).
\end{equation}
Namely, $\Lambda(t)$ is CP-divisible if and only if its generator has the standard time-dependent Gorini-Kossakowski-Sudarshan-Lindblad (GKSL) form \cite{GKS,L}
\begin{equation}
\begin{split}
\mathcal{L}(t)[\rho]=&-i[H(t),\rho]\\&+\sum_\alpha\gamma_\alpha(t)\Bigg(V_\alpha(t)\rho V_\alpha^\dagger(t)-\frac 12\{V_\alpha^\dagger(t)V_\alpha(t),\rho\}\Bigg)
\end{split}
\end{equation}
with the decoherence rates $\gamma_\alpha(t)\geq 0$. The choice of the Hamiltonian $H(t)$ as well as the noise operators $V_\alpha$, which are responsible for decoherence and dissipation phenomena, are not constrained in any means.

However, if a dynamical map is noninvertible, then $\gamma_\alpha(t)\geq 0$ is no longer necessary for CP-divisibility \cite{div_inf_flow}. In addition, the propagator itself can no longer be defined by the inverse of $\Lambda(t)$. Instead, one introduces a generalized inverse $\Lambda^-(t)$ of $\Lambda(t)$, so that $V(t,s)=\Lambda(t)\Lambda^-(s)$ \cite{Ujan}. In general, $\Lambda^-(t)$ is not uniquely defined, but for invertible maps it reduces to the standard inverse. Construction the P and CP-divisibility conditions for noninvertible dynamical maps is a relatively new research topic that requires further study \cite{div_inf_flow,Sagnik,Ujan}. It is motivated by the fact that many interesting physical evolutions are described by noninvertible dynamical maps, including a two-level atom driven by a phase noisy laser \cite{Cresser}, a damped two-level atom interacting with a single-mode field in the Jaynes-Cummings model \cite{Garraway,Wonderen,Andersson}, an amplitude damping model on resonance with Lorentzian reservoir spectrum \cite{Acin}, and a quantum NOT gate implementation \cite{NOT}. 
An important class of quantum evolution, where the system undergoes decoherence without dissipation, is provided by unital quantum maps \cite{Omkar}. The most general form of unital qubit maps is provided by the Pauli channels \cite{Landau,King}. Non-invertible Pauli dynamical maps follow for a two-level system in dissipative environment \cite{BLPV}, a two-level atom driven by a phase noisy laser \cite{Cresser}, and a two-level quantum system interacting with an environment possessing random telegraph signal noise \cite{Daffer}.

This paper is the continuation of Ref. \cite{CCMK}, where we analyzed convex combinations of non-invertible dynamical maps. For the Pauli maps, we showed that mixing non-invertible maps results in a shift, addition, or even removal of singular points. In particular, the conditions for obtaining the Markovian semigroup were presented. The main goal of this paper is to check how the divisibility properties of the maps influence the Markovianity of their mixtures. The scope is limited to bistochastic qubit evolution due to their simple spectral properties. Such analysis has so far been performed only for invertible Pauli maps, including convex combinations of Markovian semigroups \cite{ENM,Nina,Jagadish2} and CP-divisible dynamical maps \cite{CCMS,Jagadish3}. Experimental investigation of mixing two Pauli channels was performed in a photonic setup \cite{Uriri}.

In the following sections, we recall the general properties of the Pauli channels and define the Pauli dynamical maps $\Lambda(t)$ that describe the evolution of open quantum systems. Next, we use the results for image non-increasing dynamical maps to formulate the divisibility conditions for $\Lambda(t)$ that are in general noninvertible. Finally, we analyze the Markovianity properties for mixtures of non-invertible Pauli dynamical maps. We illustrate our results with several examples, showing that there are certain relations between the divisibility of maps and their convex combinations. In particular, we propose a legitimate P-divisible dynamical map generated via $\mathcal{L}(t)$ whose all decoherence rates are temporarily infinite. In Conclusions, we provide a summary of our results and a list of open questions.

\section{Pauli channels}

Consider the mixed unitary evolution of a qubit given by the Pauli channel \cite{TQI}
\begin{equation}\label{Pauli}
\Lambda[\rho]=\sum_{\alpha=0}^3 p_\alpha \sigma_\alpha \rho \sigma_\alpha,
\end{equation}
where $p_\alpha$ denotes the probability distribution and $\sigma_\alpha$ are the Pauli matrices. Mixed unitary channels arise from unitary evolution disrupted by classical errors, and they are also referred to as {\it random unitary evolution} \cite{Scheel} or {\it evolution under random external fields} \cite{Alicki}.
Notably, $\Lambda$ describes the most general bistochastic evolution of a qubit \cite{King,Landau}. Alternatively, the Pauli channels are defined via their eigenvalue equations,
\begin{equation}
\Lambda[\sigma_\alpha]=\lambda_\alpha\sigma_\alpha,\qquad \lambda_0=1,
\end{equation}
where the (real) eigenvalues $\lambda_\alpha$ relate to the probability distribution $p_\alpha$ as follows,
\begin{equation}
\lambda_\alpha = p_0 + 2p_\alpha - \sum_{\beta=1}^3 p_\beta,\qquad\alpha=1,2,3.
\end{equation}
The inverse relation reads
\begin{equation}
\begin{split}
p_0 &= \frac 14(1 + \lambda_1 + \lambda_2 + \lambda_3),\\
p_\alpha &= \frac 14 \left( 1 + 2\lambda_\alpha - \sum_{\beta=1}^3 \lambda_\beta \right),\quad \alpha=1,2,3.
\end{split}
\end{equation}
The complete positivity conditions for $\Lambda$ are given by the Fujiwara-Algoet conditions \cite{Fujiwara,King}
\begin{equation}\label{Fuji-2}
-1 \leq \sum_{\beta=1}^{3} \lambda_\beta\leq 1+2\min_{\beta}\lambda_\beta.
\end{equation}

The evolution of open quantum systems is represented by dynamical maps $\Lambda(t)$. In case of the Pauli dynamical maps, the time-dependence manifests itself in their eigenvalues $\lambda_\alpha(t)$, whereas the eigenvectors $\sigma_\alpha$ remain independent of time.
An interesting class includes mixtures of legitimate qubit dynamics. One usually considers convex combinations of phase-damping channels
\begin{equation}\label{deph}
\Lambda_\alpha(t)[\rho]=(1-p(t))\rho+p(t)\sigma_\alpha\rho\sigma_\alpha
\end{equation}
with the eigenvalue equations
\begin{equation}
\Lambda_\alpha(t)[\sigma_\alpha]=\sigma_\alpha,\qquad \Lambda_\alpha(t)[\sigma_\beta]=\lambda(t)\sigma_\beta,\quad\beta\neq\alpha,
\end{equation}
where $\lambda(t)=1-2p(t)\in[-1,1]$. The resulting map reads
\begin{equation}\label{mix}
\Lambda(t)[\rho]=\sum_{\alpha=1}^{3}x_\alpha\Lambda_\alpha(t)[\rho]=(1-p(t))\rho+p(t)
\sum_{\alpha=1}^{3} x_\alpha\sigma_\alpha\rho\sigma_\alpha,
\end{equation}
and its eigenvalues $\lambda_\alpha(t)$ are related to $\lambda(t)$ via
\begin{equation}\label{lat}
\lambda_\alpha(t)=x_\alpha+(1-x_\alpha)\lambda(t).
\end{equation}
The choice of $\lambda(t)=\exp(-rt)$ with positive $r$ corresponds to the mixtures of Markovian semigroups \cite{mub_final,ICQC}
\begin{equation}
\Lambda(t)=\sum_{\alpha=1}^{3}x_\alpha e^{rt\mathcal{L}_\alpha},
\end{equation}
where the semigroup generator
\begin{equation}\label{gen}
\mathcal{L}_\alpha[\rho]=\frac 12\left(\sigma_\alpha\rho\sigma_\alpha -\rho\right).
\end{equation}
This class includes the celebrated eternally non-Markovian evolution \cite{ENM,Nina}. A more general case has been analyzed in refs. \cite{CCMS,Jagadish3}, where the authors considered convex combinations of CP-divisible dynamical maps
\begin{equation}
\Lambda(t)=\sum_{\alpha=1}^{3}x_\alpha e^{r(t)\mathcal{L}_\alpha}.
\end{equation}
Note that these maps follow from eq. (\ref{mix}) for $\lambda(t)=\exp[-r(t)]$ with $r(t)\geq 0$. Finally, in ref. \cite{CCMK}, even more general mixtures have been analyzed, where the phase-damping channels are non-invertible (that is, $\lambda(t)=0$ for some $t>0$).

\section{Divisibility vs. indivisibility of dynamical maps}

A dynamical map $\Lambda(t)$ is divisible if and only if it can be decomposed into a trace-preserving map $V(t,s)$ and itself at an earlier time, so that $\Lambda(t)=V(t,s)\Lambda(s)$. Actually, for $\Lambda(t)$ that are invertible at all times $t\geq 0$, one always finds the propagator $V(t,s)=\Lambda(t)\Lambda^{-1}(s)$. Therefore, such maps are always divisible. Moreover, the P and CP-divisibility of invertible dynamical maps has a full mathematical characterization.

\begin{Theorem}[\cite{CKR}]\label{inv}
An invertible dynamical map $\Lambda(t)$ is P-divisible if and only if
\begin{equation}
\frac{\der}{\der t}\|\Lambda(t)[X]\|_1\leq 0
\end{equation}
and CP-divisible if and only if
\begin{equation}
\frac{\der}{\der t}\|\oper\otimes\Lambda(t)[Y]\|_1\leq 0
\end{equation}
for all Hermitian $X\in\mathcal{B}(\mathcal{H})$ and $Y\in\mathcal{B}(\mathcal{H}\otimes\mathcal{H})$, respectively.
\end{Theorem}

The problem becomes more complicated when one relaxes the invertibility condition and also allows for consideration of noninvertible maps. Recall that a map $\Lambda(t)$ is noninvertible if there exists a time $t\geq 0$ for which $\Lambda^{-1}(t)$ is not well defined. In general, dynamical maps are not necessarily divisible.

\begin{Theorem}[\cite{div_inf_flow}]\label{DIV}
A dynamical map $\Lambda(t)$ is divisible if and only if it is kernel non-decreasing; that is,
\begin{equation}\label{ker}
\bigforall_{0\leq s\leq t}\quad\ker\Lambda(s)\subseteq\ker\Lambda(t).
\end{equation}
\end{Theorem}

Note that every invertible map is divisible in a trivial manner, as $\ker\Lambda(s)=\ker\Lambda(t)=\{0\}$. For the Pauli dynamical maps, eq. (\ref{ker}) translates to the constraints on their eigenvalues only, as will be seen later.

Now, consider the class of image non-increasing dynamical maps that satisfy
\begin{equation}\label{im}
\bigforall_{0\leq s\leq t}\quad\mathrm{Im}\Lambda(t)\subseteq
\mathrm{Im}\Lambda(s).
\end{equation}
For such maps, there exists a generalization of Theorem \ref{inv}.

\begin{Theorem}[\cite{div_inf_flow}]\label{ini}
An image non-increasing dynamical map $\Lambda(t)$ is P-divisible if
\begin{equation}
\frac{\der}{\der t}\|\Lambda(t)[X]\|_1\leq 0
\end{equation}
and CP-divisible if
\begin{equation}
\frac{\der}{\der t}\|\oper\otimes\Lambda(t)[Y]\|_1\leq 0
\end{equation}
for every Hermitian $X\in\mathcal{B}(\mathcal{H})$ and $Y\in\mathcal{B}(\mathcal{H}\otimes\mathcal{H})$, respectively.
\end{Theorem}

For qubit maps, Theorem \ref{ini} provides the necessary and sufficient condition for CP-divisibility \cite{Sagnik} but not P-divisibility.

Note that not all kernel non-decreasing maps are image non-increasing, as eq. (\ref{ker}) does not imply eq. (\ref{im}) but instead $\dim\mathrm{Im}\Lambda(t)\leq\dim\mathrm{Im}\Lambda(s)$. There are two natural classes of image non-increasing maps \cite{div_inf_flow}:
\begin{itemize}
\item divisible normal maps; a map $\Lambda(t)$ is by definition normal if $\Lambda^\dagger(t)\Lambda(t)=\Lambda(t)\Lambda^\dagger(t)$, where $\Lambda^\dagger(t)$ is the dual map obtained via $\Tr(X^\dagger\Lambda(t)[Y])=\Tr(\Lambda^\dagger[X^\dagger]Y)$;
\item commutative maps ($\Lambda(t)\Lambda(s)=\Lambda(s)\Lambda(t)$ for any $t,s\geq 0$) that are diagonalizable; a diagonal representation of $\Lambda(t)$ is given by
\begin{equation}
\Lambda(t)[X]=\sum_{\alpha=0}^{d^2-1}\lambda_\alpha(t)F_\alpha\Tr(G_\alpha^\dagger X)
\end{equation}
with a basis of $\Lambda$'s eigenvectors $F_\alpha$ and its dual basis $G_\alpha$ ($\Tr(F_\alpha^\dagger G_\beta)=\delta_{\alpha\beta}$).
\end{itemize}
It turns out that divisible Pauli dynamical maps belong to both of the aforementioned categories. Indeed, it is straightforward to show that $\Lambda(t)$ given by a time-dependent version of eq. (\ref{Pauli}) is self-dual ($\Lambda^\dagger(t)=\Lambda(t)$) and hence normal. Additionally, it is commutative due to
\begin{equation}
\Lambda(t)\Lambda(s)[\sigma_\alpha]=\lambda_\alpha(s)\lambda_\alpha(t)[\sigma_\alpha]
=\Lambda(s)\Lambda(t)[\sigma_\alpha].
\end{equation}
Finally, $\Lambda(t)$ is unital, and so $G_\alpha=F_\alpha=\sigma_\alpha$. Therefore, it can also be rewritten in the diagonal form
\begin{equation}
\Lambda(t)[X]=\Phi_0[X]+\sum_{\alpha=1}^{3}\lambda_\alpha(t)
\sigma_\alpha\Tr(\sigma_\alpha X),
\end{equation}
where $\Phi_0[X]=\frac 12 \mathbb{I}\Tr X$ is the completely depolarizing channel.

For the bistochastic qubit evolution, we show that divisibility can be determined using only the eigenvalues $\lambda_\alpha(t)$ of the corresponding Pauli channel.

\begin{Theorem}\label{PC}
Any legitimate Pauli dynamical map $\Lambda(t)$ is
\begin{enumerate}[label={\bf(\arabic*)}]
\item divisible if and only if $\lambda_\alpha(t)\geq 0$ and
\begin{equation}
\lambda_\alpha(s)=0\qquad\implies\qquad\lambda_\alpha(t\geq s)=0;
\end{equation}
\item P-divisible if and only if it is divisible and
\begin{equation}
\dot{\lambda}_\alpha(t)\leq 0;
\end{equation}
\item CP-divisible if and only if it is P-divisible, satisfies
\begin{equation}\label{gamma}
2\frac{\der}{\der t}\ln[\lambda_\alpha(t)]\geq
\frac{\der}{\der t}\ln[\lambda_1(t)\lambda_2(t)\lambda_3(t)]
\end{equation}
whenever all $\lambda_\alpha(t)>0$, and never has exactly two non-zero eigenvalues $\lambda_\alpha(t)$ ($\alpha=1,2,3$) at any time $t\geq 0$.
\end{enumerate}
\end{Theorem}

The proof can be found in the Appendix. Note that the conditions for divisibility and CP-divisibility follow from the general results in Theorems \ref{DIV} and \ref{ini}. However, the necessary and sufficient conditions for P-divisibility are not known even for the qubit maps, and therefore their derivation requires more attention.

As an example of an indivisible dynamical map, consider the depolarizing channel
\begin{equation}
\Lambda(t)[\rho]=(1-p(t))\rho+\frac{p(t)}{3}\sum_{\alpha=1}^{3}\sigma_\alpha
\rho\sigma_\alpha
\end{equation}
with the eigenvalues
\begin{equation}
\lambda_\alpha(t)=1-\frac{4}{3}p(t).
\end{equation}
Note that this map is completely positive if and only if $-1/3\leq\lambda_\alpha(t)\leq 1$. Hence, if we choose $\lambda_\alpha(t)=|\cos\omega t|$, then, by point (1) in Theorem \ref{PC}, $\Lambda(t)$ is indivisible. Indeed, $\lambda_\alpha(\pi/2\omega)=0$ does not imply $\lambda_\alpha(t)=0$ for all $t>\pi/2\omega$. However, this can be easily remedied if one modifies the eigenvalues, so that
\begin{equation}
\lambda_\alpha(t)=\left\{\begin{aligned}
\cos\omega t,\qquad &t\leq\frac{\pi}{2\omega},\\
0,\qquad &t>\frac{\pi}{2\omega}.
\end{aligned}\right.
\end{equation}
The corresponding depolarizing channel is indeed divisible.

An instructive example of the Markovian qubit evolution described by a noninvertible Pauli dynamical map was proposed in ref. \cite{div_inf_flow}. Let us take the qubit dynamical map generated via
\begin{equation}\label{gen_PC}
\mathcal{L}(t)=\sum_{\alpha=1}^3\gamma_\alpha(t)\mathcal{L}_\alpha,\qquad
\mathcal{L}_\alpha[X]=\frac 12 (\sigma_\alpha X\sigma_\alpha-X).
\end{equation}
The solution of $\dot{\Lambda}(t)=\mathcal{L}(t)\Lambda(t)$ is the Pauli dynamical map with eigenvalues
\begin{equation}\label{eigenv_PC}
\lambda_\alpha(t)=\exp[\Gamma_\alpha(t)-\Gamma_0(t)],
\end{equation}
where $\gamma_0=\sum_{\alpha=1}^3\gamma_\alpha$ and $\Gamma_\alpha(t)=\int_0^t\gamma_\alpha(\tau)\der\tau$. Note that $\Lambda(t)$ is invertible if and only if $\Gamma_\alpha(t)$ are all finite for finite times $t\geq 0$. Now, if $\Gamma_3(t_1)=\infty$ at a finite time $t_1$, then the eigenvalues of the dynamical map $\lambda_1(t_1)=\lambda_2(t_1)=0$. Also, the divisibility property implies that $\lambda_1(t)=\lambda_2(t)=0$ for $t\geq t_1$. If there is a time $t_2>t_1$ such that $\Gamma_2(t_2)=\infty$, then one has $\lambda_3(t)=0$ for any $t\geq t_2$. From now on, $\Gamma_1(t)$ can be arbitrary, and the system stays in the maximally mixed state. Therefore, one can conclude the general property for time-local generators. Namely, a Pauli dynamical map is CP-divisible if and only if all $\gamma_\alpha(t)\geq 0$ up to a time $t<t_1$, and $\gamma_{\alpha_\ast}(t)\geq 0$ for one fixed $\alpha_\ast$ until $t<t_2$.

\section{Mixtures of Pauli dynamical maps}

A special class of evolution is given by convex combinations of legitimate quantum dynamical maps. In our previous work \cite{CCMK}, we analyzed the behaviour of the singular points of the dynamical maps after mixing, but we completely omitted the discussion of memory effects. In this paper, our main goal is to find the relation between the divisibility (and hence degrees of non-Markovianity) of the phase-damping channels $\Lambda_\alpha(t)$ given in eq. (\ref{deph}) and their mixtures. All the results are general unless stated otherwise.

First, let us analyze the properties of $\Lambda_\alpha(t)$, which are a trivial case of mixing with $x_\alpha=1$ and $x_\beta=0$ for $\beta\neq\alpha$. Obviously, these maps are invertible if and only if $\lambda(t)\neq 0$ at all $t\geq 0$. Now, from Theorem \ref{PC}, we see that $\Lambda_\alpha(t)$ are indivisible if and only if the eigenvalue $\lambda(t)$ does not stay equal to zero after vanishing at some point in time. Good examples of $\lambda(t)$ corresponding to indivisible $\Lambda_\alpha(t)$ are
\begin{align}
\lambda(t)=&\cos\omega t,\\
\lambda(t)=&|\cos\omega t|,\\
\lambda(t)=&e^{-Zt}\cos\omega t,
\end{align}
with $Z,\omega>0$ \cite{CCMK}. Due to the maps having only one free parameter $\lambda(t)$, there are no $\Lambda_\alpha(t)$ that are P-divisible but not CP-divisible. However, one can still produce interesting examples for divisible but not CP-divisible maps after fixing a nonmonotonic function $\lambda(t)$.

\begin{Example}\label{EXdiv}
A divisible but not CP-divisible $\Lambda_\alpha(t)$ follows from
\begin{equation}
\lambda(t)=\left\{
\begin{aligned}
\frac{1}{c}\left(1-\frac{t}{T}\right)\left[a\left(\frac{t}{T}\right)^2
+b\frac{t}{T}+c\right],&\quad t<T,\\
0,&\quad t\geq T,
\end{aligned}\right.
\end{equation}
where the parameters $a$, $b$, $c$ are chosen in such a way that
\begin{equation}
0<b<a,\qquad \frac{(a+b)^2}{4a}\leq c<\frac{a^2+ab+b^2}{3a}.
\end{equation}
\end{Example}

Finally, from Corollary \ref{PC}, we see that CP-divisible phase-damping channels $\Lambda_\alpha(t)$ correspond to monotonically decreasing $\lambda(t)$, or equivalently $\dot{\lambda}(t)\leq 0$.

\begin{Example}
One gets a CP-divisible $\Lambda_\alpha(t)$ for
\begin{equation}
\lambda(t)=\left\{\begin{aligned}
\cos\omega t,&\quad t\leq\frac{\pi}{2\omega},\\
0,&\quad t>\frac{\pi}{2\omega}.
\end{aligned}\right.
\end{equation}
\end{Example}

This way, we provide a full characterization of the divisibility properties for $\Lambda_\alpha(t)$. Now, we move our attention to their convex combinations $\Lambda(t)=\sum_{\alpha=1}^{3}x_\alpha\Lambda_\alpha(t)$. In ref. \cite{CCMK}, it was proven that invertible maps do not produce noninvertible mixtures. Moreover, the non-Markovianity of $\Lambda(t)$ for invertible $\Lambda_\alpha(t)$ has already been analyzed in ref. \cite{CCMS}. Therefore, our main focus and original results are the Pauli dynamical maps obtained from noninvertible phase-damping channels.

Let us start with the mixtures of $\Lambda_\alpha(t)$ that satisfy the most restrictive conditions, and then move to the more general choices. This way, we are going to observe that the less constraints is put on $\Lambda_\alpha(t)$, the wider classes of dynamical maps follow from taking their convex combinations. First, let us provide a generalization of the results for mixtures of CP-divisible invertible Pauli maps from ref. \cite{CCMS}.

\begin{Proposition}\label{CPP}
Mixtures of CP-divisible $\Lambda_\alpha(t)$ always produce P-divisible $\Lambda(t)$.
\end{Proposition}

\begin{proof}
Recall that the eigenvalues of $\Lambda(t)$ are given by $\lambda_\alpha(t)=x_\alpha+(1-x_\alpha)\lambda(t)$. Therefore, the P-divisibility condition for $\Lambda(t)$, which is
\begin{equation}\label{dot}
\dot{\lambda}_\alpha(t)=(1-x_\alpha)\dot{\lambda}(t)\leq 0,
\end{equation}
holds if and only if $\dot{\lambda}(t)\leq 0$. On the other hand, $\dot{\lambda}(t)\leq 0$ is only necessary for CP-divisibility of $\Lambda(t)$.
\end{proof}

Note that if all $x_\alpha\neq 0$, then the resulting $\Lambda(t)$ is invertible. Now, one can also use CP-divisible $\Lambda_\alpha(t)$ to produce CP-divisible $\Lambda(t)$.

\begin{Proposition}\label{CPCP}
Mixtures of CP-divisible $\Lambda_\alpha(t)$ produce CP-divisible $\Lambda(t)$ if and only if $x_\alpha\neq 0$ for $\alpha=1,2,3$ and
\begin{equation}\label{gamma2}
\lambda(t)\geq\frac{1}{1-x_k}\left(-x_k+\sqrt{\frac{(x_i-x_k)(x_j-x_k)}{(1-x_i)(1-x_j)}}\right),
\end{equation}
where $\{x_i,x_j,x_k\}=\{x_1,x_2,x_3\}$.
The resulting maps $\Lambda(t)$ are always invertible.
\end{Proposition}

\begin{proof}
First, let us show that $\Lambda(t)$ is always invertible. We know that mixtures of invertible dynamical maps always produce invertible dynamical maps \cite{CCMK}. For noninvertible $\Lambda_\alpha(t)$, there is a time $t_\ast\geq 0$ for which $\lambda(t_\ast)=0$. Now, $\Lambda_\alpha(t)$ are divisible, so at any $t\geq t_\ast$ the eigenvalues of the mixture are $\lambda_\alpha(t)=x_\alpha+(1-x_\alpha)\lambda(t_\ast)=x_\alpha$. If $\Lambda(t)$ is CP-divisible and noninvertible, there have to be two $x_\alpha=0$. However,  this case corresponds to $\Lambda(t)=\Lambda_\alpha(t)$, so there is no mixing. Hence, no CP-divisible noninvertible maps result from our mixtures.

Now, if $\Lambda(t)$ is invertible, then its CP-divisibility condition reads
\begin{equation}
\frac{2x_\alpha}{x_\alpha+(1-x_\alpha)\lambda(t)}\geq
\sum_{\beta=1}^3\frac{x_\beta}{x_\beta+(1-x_\beta)\lambda(t)}-1,
\end{equation}
which is a direct consequence of condition (\ref{gamma}), $\lambda(t)\geq 0$, and $\dot{\lambda}(t)\leq 0$. After eliminating the fractions, one gets
\begin{equation}\label{ineq}
\begin{split}
\lambda\Big[&\lambda^2(1-x_i)(1-x_j)(1-x_k)+2\lambda(1-x_i)(1-x_j)x_k\\&-x_ix_j+x_jx_k
+x_ix_k-x_ix_jx_k\Big]\geq 0,
\end{split}
\end{equation}
where $\{x_i,x_j,x_k\}=\{x_1,x_2,x_3\}$.
The solutions of the quadratic equation are
\begin{equation}
\lambda_\pm(t)=\frac{1}{1-x_k}\left(-x_k\pm\sqrt{\frac{(x_i-x_k)(x_j-x_k)}{(1-x_i)(1-x_j)}}\right).
\end{equation}
Hence, eq. (\ref{ineq}) reduces to
\begin{equation}
\lambda(\lambda-\lambda_+)(\lambda-\lambda_-)\geq 0.
\end{equation}
Observe that $\lambda(t)\geq 0$ and $\lambda_-(t)\leq 0$, and therefore the above inequality holds for $\lambda(t)\geq\lambda_+(t)$, which is exactly eq. (\ref{gamma2}).
\end{proof}

If $\min_{t\geq 0}\lambda(t)=0$, then eq. (\ref{gamma2}) reduces to the constraint that involves only $x_\alpha$,
\begin{equation}
\frac{1}{1-x_k}\left(-x_k+\sqrt{\frac{(x_i-x_k)(x_j-x_k)}{(1-x_i)(1-x_j)}}\right)\leq 0.
\end{equation}
This is equivalent to
\begin{equation}
x_ix_jx_k+x_ix_j-x_jx_k-x_ix_k\leq 0,
\end{equation}
which is exactly the CP-divisibility condition for mixtures of Markovian semigroups obtained in \cite{Nina}. Therefore, Propositions \ref{CPP} and \ref{CPCP} show that the statements known for the convex combinations of Markovian semigroups generalize to combinations of any legitimate Pauli dynamical maps.

\begin{Example}
For $x_\alpha=1/3$, mixtures $\Lambda(t)$ of CP-divisible $\Lambda_\alpha(t)$ are also CP-divisible.
\end{Example}

Proposition \ref{CPCP} also provides a generalization of the theorem in ref. \cite{Jagadish2} which states that a mixture of two Markovian Pauli semigroups always results in a dynamical map that is not CP-divisible.

\begin{Example}
Let us take $x_1=x_2=1/2$, $x_3=0$, and
\begin{equation}
\lambda(t)=\left\{\begin{aligned}
\cos\omega t,&\quad t\leq\frac{\pi}{2\omega},\\
0,&\quad t>\frac{\pi}{2\omega}.
\end{aligned}\right.
\end{equation}
The resulting map $\Lambda(t)$ possesses the following eigenvalues,
\begin{equation}
\begin{split}
\lambda_1(t)=\lambda_2(t)=&\left\{\begin{aligned}
\frac{1+\cos\omega t}{2},&\quad t\leq\frac{\pi}{2\omega},\\
\frac 12,&\quad t>\frac{\pi}{2\omega},\qquad
\end{aligned}\right.\\
\lambda_3(t)=&\left\{\begin{aligned}
\cos\omega t,&\quad t\leq\frac{\pi}{2\omega},\\
0,&\quad t>\frac{\pi}{2\omega}.\qquad
\end{aligned}\right.
\end{split}
\end{equation}
Observe that $\Lambda(t)$ is not CP-divisible, as these is a time for which it has two non-zero eigenvalues. However, it is P-divisible due to $\dot{\lambda}_\alpha(t)\leq 0$.
\end{Example}

The dynamical map $\Lambda(t)$ from the above example can be generated using a time-local generator $\mathcal{L}(t)$ given by eq. (\ref{gen}) with the decoherence rates
\begin{equation}
\begin{split}
\gamma_1(t)=\gamma_2(t)=&\left\{\begin{aligned}
\frac{\omega}{2}\tan(\omega t),&\quad t\leq\frac{\pi}{2\omega},\\
0,&\quad t>\frac{\pi}{2\omega},\qquad
\end{aligned}\right.\\
\gamma_3(t)=&\left\{\begin{aligned}
-\frac{\omega}{2}\tan(\omega t)\frac{1-\cos\omega t}{1+\cos\omega t},&\quad t\leq\frac{\pi}{2\omega},\\
0,&\quad t>\frac{\pi}{2\omega},\qquad
\end{aligned}\right.
\end{split}
\end{equation}
Interestingly, despite the rates being infinite at one point,
\begin{equation}
\gamma_1\left(\frac{\pi}{2\omega}\right)=\gamma_2\left(\frac{\pi}{2\omega}\right)
=\infty,\qquad \gamma_3\left(\frac{\pi}{2\omega}\right)=-\infty,
\end{equation}
their sums are finite,
\begin{equation}
\gamma_1\left(\frac{\pi}{2\omega}\right)+\gamma_3\left(\frac{\pi}{2\omega}\right)
=\gamma_2\left(\frac{\pi}{2\omega}\right)+\gamma_3\left(\frac{\pi}{2\omega}\right)
=\omega.
\end{equation}
This provides the first example of the time-local generator for a PnCP-divisible dynamical map that is not invertible.

Next, we observe what happens if one mixes the Pauli maps $\Lambda_\alpha(t)$ that are not even P-divisible. Such maps satisfy $\lambda(t_\ast)=0$ $\implies$ $\lambda(t\geq t_\ast)=0$ but violate condition $\dot{\lambda}(t)\leq 0$. In eq. (\ref{dot}), we have already shown that $\dot{\lambda}(t)\leq 0$ if and only if $\dot{\lambda}_\alpha(t)\leq 0$, which is the P-divisibility condition for $\Lambda(t)$. Therefore, we prove the following proposition.

\begin{Proposition}
Convex combinations of divisible but not P-divisible $\Lambda_\alpha(t)$ lead to divisible but not P-divisible $\Lambda(t)$.
\end{Proposition}

If every $x_\alpha\neq 0$, then the mixtures of divisible but not P-divisible $\Lambda_\alpha(t)$ are always invertible.

There are also some interesting properties for mixtures of indivisible phase-damping channels. It was shown in \cite{CCMK} that such mixtures can produce CP-divisible maps. In this case, $\lambda(t)$ is monotonically decreasing, even though it passes through zero. Actually, $\dot{\lambda}(t)\leq 0$ is a necessary condition for a mixture of $\Lambda_\alpha(t)$ to be P-divisible. If we want to guarantee the divisibility of $\Lambda(t)$, then it turns out that there are easy-to-check necessary and sufficient conditions.

\begin{Proposition}
Mixing indivisible $\Lambda_\alpha(t)$ produces divisible $\Lambda(t)$ if and only if
\begin{equation}\label{cc}
\lambda(t_\ast)=-\frac{x_{\min}}{1-x_{\min}}\quad\implies\quad
\lambda(t)=-\frac{x_{\min}}{1-x_{\min}}
\end{equation}
for $t\geq t_\ast$ and $x_{\min}=\min_\alpha x_\alpha$.
\end{Proposition}

The above constraint is a direct consequence of point (1) in Theorem \ref{PC}. Whenever condition (\ref{cc}) is violated, the mixture of indivisible maps remains indivisible. Recall that all indivisible Pauli dynamical maps are noninvertible.

\section{Conclusions}

In this paper, we analyzed the divisibility of Pauli dynamical maps that are in general noninvertible. We provided the necessary and sufficient conditions for these maps to be divisible, P-divisible, and CP-divisible in terms of their eigenvalues alone. Next, we considered the qubit evolution given by legitimate mixtures of Pauli dynamical maps. We found the following interesting connections between the divisibility properties of the most general phase-damping qubit maps and their convex combinations.
\begin{itemize}
\item Indivisible maps follow only from mixing indivisible maps.
\item Mixtures of divisible maps always produce divisible maps.
\item Convex combinations of CP-divisible maps are at least P-divisible.
\item Noninvertible CP-divisible maps cannot be obtained through mixing.
\item Time-local generators of noninvertible dynamical maps have decoherence rates that become infinite for finite times.
\item Mixtures of CP-divisible maps are at least P-divisible.
\item Indivisible dynamical maps can be mixed into maps with any degree of divisibility.
\end{itemize}
We also provided examples of time-local generators for noninvertible dynamical maps, including a legitimate generator with all decoherence rates infinite at a finite time $t>0$.

There are still many open questions regarding the divisibility of dynamical maps. Due to the complexity of calculating the trace norms for high-dimensional systems, we were unable to provide a satisfactory analysis for the generalized Pauli channels. It would be interesting to analyze the divisibility of mixtures of more involved Pauli maps \cite{Jagadish4} and noninvertible generalized Pauli dynamical maps. One could even find relations between $k$-divisibility of channels \cite{k-pos} and their convex combinations by considering something more general than phase-damping maps.

\section{Acknowledgements}

This paper was supported by the Foundation for Polish Science (FNP) and the Polish National Science Centre project No. 2018/30/A/ST2/00837. The author thanks Dariusz Chru\'{s}ci\'{n}ski for helpful discussions.

\bibliography{C:/Users/cynda/OneDrive/Fizyka/bibliography}
\bibliographystyle{C:/Users/cynda/OneDrive/Fizyka/beztytulow2}

\appendix

\section{Proof to Theorem \ref{PC}}

Any Pauli dynamical map $\Lambda(t)$ is divisible if and only if $\lambda_\alpha(t)\geq 0$ and
\begin{equation}
\lambda_\alpha(s)=0\qquad\implies\qquad\lambda_\alpha(t\geq s)=0,
\end{equation}
which is a direct consequence of Theorem \ref{DIV}.
Assume the most general form of a divisible Pauli dynamical map -- namely, let the eigenvalues $\{\lambda_1(t),\lambda_2(t),\lambda_3(t)\}=\{\lambda_i(t),\lambda_j(t),\lambda_k(t)\}$ evolve in such a way that
\[
\begin{split}
0\leq t<t_1:&
\qquad(\lambda_i(t),\lambda_j(t),\lambda_k(t))=(\lambda(t),\eta(t),\mu(t)),\\
t_1\leq t<t_2:&
\qquad(\lambda_i(t),\lambda_j(t),\lambda_k(t))=(\lambda(t),\eta(t),0),\\
t_2\leq t<t_3:&
\qquad(\lambda_i(t),\lambda_j(t),\lambda_k(t))=(\lambda(t),0,0),\\
t\geq t_3:&
\qquad(\lambda_i(t),\lambda_j(t),\lambda_k(t))=(0,0,0),
\end{split}
\]
where $\lambda(t),\eta(t),\mu(t)>0$ are arbitrary functions for which $\Lambda(t)$ is completely positive.
In other words, $\Lambda(t)$ has three non-zero eigenvalues up to a time $t_1$. Then, at $t=t_1$, one of the eigenvalues vanishes. As the dynamical map is divisible, any vanishing eigenvalues will not reappear at later times. Next, the two remaining eigenvalues evolve up to $t=t_2\geq t_1$, when another eigenvalue vanishes. The final eigenvalue reaches zero at $t=t_3\geq t_2$. From then on, the system stays at the maximally mixed state. Note that we allow $t_j=t_k$ for $j,k=1,2,3$, and even $t_k\to\infty$, which would mean that some of the eigenvalues do not vanish for finite times.

Now, consider the problem of finding the P and CP-divisibility conditions for $\Lambda(t)$ separately for each range of time. If $\Lambda(t)$ satisfies the P or CP-divisibility conditions at all times $t\geq 0$, then it is P or CP-divisible, respectively. Without a loss of generality, we assume that $\Lambda(t)$ is a solution of the master equation $\dot{\Lambda}(t)=\mathcal{L}(t)\Lambda(t)$ with a time-local generator defined as in eq. (\ref{gen_PC}).

First, we observe that for $0\leq t<t_1$ the dynamical map is invertible, and hence the divisibility conditions coincide with the conditions for invertible maps. Recall that invertible Pauli dynamical maps are P-divisible if and only if 
\begin{equation}
\dot{\lambda}_\alpha(t)\leq 0
\end{equation}
and CP-divisible if and only if its time-local generator $\mathcal{L}(t)$ has all its decoherence rates $\gamma_\alpha(t)\geq 0$ \cite{Filip}. Note that the condition for the generator is equivalent to
\begin{equation}
2\frac{\der}{\der t}\ln[\lambda_\alpha(t)]\geq
\frac{\der}{\der t}\ln[\lambda_1(t)\lambda_2(t)\lambda_3(t)]
\end{equation}
due to eq. (\ref{eigenv_PC}) that connects $\lambda_\alpha(t)$ with $\gamma_\alpha(t)$.

Now, we still need to formulate the conditions for the times at which $\Lambda^{-1}(t)$ is not well defined. For this purpose, let us recall the divisibility properties for the Pauli channels $\Lambda$. From definition, any quantum channel $\Lambda$ is divisible if it can be decomposed into $\Lambda=V\Lambda^\prime$, where $\Lambda^\prime$ is also a quantum channel, and $V$ is a positive, trace-preserving, non-unitary map.

\begin{Lemma}\label{nec}
Noninvertible Pauli channels $\Lambda$ are P-divisible if and only if
\begin{equation}
\lambda_1\lambda_2\lambda_3\geq 0
\end{equation}
(Theorem 25 in \cite{Cirac} or Theorem 2 in \cite{Davalos})
and CP-divisible if and only if there is only one $k\in\{1,2,3\}$ such that
\begin{equation}
\lambda_k\neq 0
\end{equation}
(Theorem 24 in \cite{Cirac}).
\end{Lemma}

This lemma provides us with necessary conditions for P and CP-divisibility of Pauli dynamical maps $\Lambda(t)$. Indeed, if $\Lambda(t)$ is P-divisible, then $\Lambda(t)=V(t,s)\Lambda(s)$ for $s\leq t$ with a positive propagator $V(t,s)$. For fixed $t$ and $s$, $\Lambda(t)$ and $\Lambda(s)$ are two Pauli channels -- denote them by $\Lambda^\prime$ and $\Lambda$, respectively. Analogically, for fixed times, $V(t,s)=V$ is a positive map. Therefore, one has $\Lambda^\prime=V\Lambda$, which is exactly the P-divisibility property for the Pauli channel $\Lambda^\prime$. Analogical reasoning follows for CP-divisible dynamical maps.

\begin{Corollary}
If a Pauli dynamical map $\Lambda(t)$ is CP-divisible, then for $t_1\leq t<t_3$ its eigenvalues are
\begin{equation}
(\lambda_i(t),\lambda_j(t),\lambda_k(t))=(\lambda(t),0,0).
\end{equation}
\end{Corollary}

In other words, for CP-divisible $\Lambda(t)$, one has $t_2=t_3$. Now, to see that $t_3$ can in general be finite, recall that for $t_1\leq t<t_3$ the corresponding propagator $V(t,s)$ defined by
\begin{equation}
\begin{split}
V(t,s)[\mathbb{I}]&=\mathbb{I},\\
V(t,s)[\sigma_i]&=\frac{\lambda(t)}{\lambda(s)}\sigma_i,\\
V(t,s)[\sigma_j]&=V(t,s)[\sigma_k]=0
\end{split}
\end{equation}
is completely positive and trace-preserving for any monotonically decreasing $\lambda(t)>0$ \cite{Ujan}. This way, we prove the second part of Theorem \ref{PC}.

To prove the first part, we consider the remaining time range, $t_1\leq t<t_2$, when there are two non-zero eigenvalues and $\Lambda(t)$ can be at most P-divisible. The associated propagator $V(t,s)$ satisfies the eigenvalue equations
\begin{equation}
\begin{split}
V(t,s)[\mathbb{I}]&=\mathbb{I},\\
V(t,s)[\sigma_i]&=\frac{\lambda(t)}{\lambda(s)}\sigma_i,\\
V(t,s)[\sigma_j]&=\frac{\eta(t)}{\eta(s)}\sigma_j,\qquad V(t,s)[\sigma_k]=0.
\end{split}
\end{equation}
Now, $V(t,s)$ is a Pauli map, so it is positive if and only if \cite{Zyczkowski}
\begin{equation}
\left|\frac{\lambda(t)}{\lambda(s)}\right|\leq 1,\qquad \left|\frac{\eta(t)}{\eta(s)}\right|\leq 1.
\end{equation}
As the functions $\lambda(t)$ and $\eta(t)$ are positive, the above condition is equivalent to a monotonical decrease of $\lambda(t)$ and $\eta(t)$. Therefore, for P-divisible $\Lambda(t)$, one has $\dot{\lambda}_\alpha(t)\leq 0$ at all times $t$.

\end{document}